\documentclass[11 pt]{article}
\usepackage[utf8]{inputenc}

\pdfoutput=1

\usepackage[table,xcdraw]{xcolor}
\usepackage{amsmath, amsthm, amssymb}
\usepackage{algpseudocode,algorithm,algorithmicx}
\floatname{algorithm}{Alg}
\usepackage{mathtools}
\usepackage[numbers]{natbib}
\usepackage{comment}
\usepackage[most]{tcolorbox}
\usepackage{xfrac}
\usepackage{hyperref}
\usepackage{multirow}
\usepackage{caption}
\usepackage{bm}
\usepackage{newfloat}
\usepackage{enumitem}
\usepackage{bbm}
\usepackage{soul}
\usepackage{makecell}
\usepackage{wrapfig}
\usepackage{hhline}
\usepackage{tikz}
\usepackage{subcaption}
\usepackage{cleveref}

\usetikzlibrary{decorations.markings}
\usetikzlibrary{positioning, shapes.geometric}
\usetikzlibrary{calc}

\makeatletter
\renewcommand{\fnum@algorithm}{Alg$_{\thealgorithm}$}
\makeatother
\crefformat{algorithm}{#2\textup{\ensuremath{\mathrm{Alg}_{#1}}}#3}
\Crefformat{algorithm}{#2\textup{\ensuremath{\mathrm{Alg}_{#1}}}#3}
\usepackage[margin=1in]{geometry}

\title{Graphic Matroid Secretary without the Graph}
\author{Paul D\"utting\thanks{Google Research, Z\"urich, Switzerland. Email: \texttt{duetting@google.com}} \and Renato Paes Leme\thanks{Google Research, New York, NY, USA. Email: \texttt{renatoppl@google.com}} \and Martin P\'{a}l\thanks{Google, Mountain View, CA, USA. Email: \texttt{mpal@google.com}} \and Neel Patel\thanks{Google Research, Z\"urich, Switzerland. Email: \texttt{neelbpatel@google.com}}}
\date{}

\usepackage{aliascnt}

\newtheorem{theorem}{Theorem}[section]

\newaliascnt{corollary}{theorem}

\aliascntresetthe{corollary}

\newaliascnt{lemma}{theorem}
\newtheorem{lemma}[lemma]{Lemma}
\aliascntresetthe{lemma}

\newaliascnt{proposition}{theorem}

\aliascntresetthe{proposition}

\newaliascnt{claim}{theorem}
\newtheorem{claim}[claim]{Claim}
\aliascntresetthe{claim}

\newaliascnt{definition}{theorem}
\newtheorem{definition}[definition]{Definition}
\aliascntresetthe{definition}

\newaliascnt{question}{theorem}

\aliascntresetthe{question}

\newaliascnt{observation}{theorem}

\aliascntresetthe{observation}

\newaliascnt{fact}{theorem}

\aliascntresetthe{fact}

\newaliascnt{example}{theorem}

\aliascntresetthe{example}

\newaliascnt{note}{theorem}

\aliascntresetthe{note}

\newaliascnt{remark}{theorem}
\newtheorem{remark}[remark]{Remark}
\aliascntresetthe{remark}

\newcommand{\abs}[1]{\left| #1 \right|}

\renewcommand{\hat}{\widehat}

\newcommand{\loops}{L}

\DeclareMathOperator{\cl}{cl}
\DeclareMathOperator{\rank}{rank}
\DeclareMathOperator{\rk}{rank}

\DeclareMathOperator{\Circ}{Circ}
\DeclareMathOperator{\Opt}{Opt}
\DeclareMathOperator{\Alg}{Alg}

\DeclareMathOperator{\argmax}{argmax}

\let\vec\mathbf

\def\EH{E_H}
\def\EW{E_W}

\def\E{\mathbb{E}}

\def\F{\mathbb{F}}

\def\I{\mathcal{I}}
\def\M{\mathcal{M}}

\def\P{\mathcal{P}}

\def\Z{\mathcal{Z}}

\algnewcommand{\IIf}[1]{\State\algorithmicif\ #1\ \algorithmicthen}
\algnewcommand{\EndIIf}{\unskip\ \algorithmicend\ \algorithmicif}

\newcounter{proc}

\newenvironment{tbox}{
\vspace{0.2cm}
\begin{tcolorbox}[width=\textwidth,
                  enhanced,
                  boxsep=2pt,
                  left=1pt,
                  right=1pt,
                  top=4pt,
                  boxrule=1pt,
                  arc=0pt,
                  colback=white,
                  colframe=black,
                  breakable]
}{
\end{tcolorbox}
}

\newcommand{\tboxhrule}[0]{\vspace{0.1cm} \hrule \vspace{0.2cm}}

\newenvironment{titledtbox}[1]{\begin{tbox}#1 \tboxhrule}{\end{tbox}}

\newcolumntype{?}{!{\vrule width 1.3pt}}

\begin{document}

\maketitle

\begin{abstract}
    The matroid secretary problem (MSP) is one of the cleanest, and most captivating open problems in online algorithms. The famous MSP conjecture stipulates that there exists a constant-competitive algorithm, yet to date the best known algorithms are $O(\log \log (\text{rank}))$ competitive. It is widely believed that all information that an algorithm for the MSP should use is information that is available through an independence oracle on the already arrived elements.
    Despite this, there are natural classes of matroids where a constant-competitive algorithm is known if we are given additional upfront information about the matroid; while no such algorithm is known if all the algorithm can use is an independence oracle on the arrived elements.
    In this work, we tackle the perhaps most appealing such class of matroids, graphic matroids. We develop an algorithm for the MSP that has access to the independence oracle only. Our algorithm runs in polynomial time, and if the underlying matroid is graphic, it produces an independent set whose weight is at least $1/36$ of the maximum-weight independent set. Ours is the first constant-competitive algorithm for MSP on unknown graphic matroids. 
\end{abstract}

\section{Introduction}

Introduced by Babaioff et al.~\cite{babaioff2007matroids} in 2007, the matroid secretary problem (MSP) is a cornerstone of online algorithms, and the MSP conjecture—which postulates the existence of a constant-competitive algorithm—remains one of the field's most prominent open questions. Formally, in the MSP we are given a matroid $\M = (E,\I)$ with elements of unknown weights $w: E \rightarrow \mathbb{R}_{\geq 0}$ chosen by an \emph{adversary}, arriving in a uniformly random order. Upon arrival, the element reveals its weight and the algorithm has to immediately and irrevocably decide whether to select it. The goal is to select a set of feasible elements $I$ in the matroid, i.e., $I \in \I$, such that the total weight of the elements $w(I) = \sum_{e \in I} w(e)$ is as large as possible.

The best-known algorithms for the MSP are $O(\log \log(\rank(\M)))$-competitive \cite{Lachish14,FeldmanSZ18}, where $\rank(\M)$ is the \emph{rank} of the matroid $\M$, i.e., the cardinality of a largest independent set. Whereas the MSP conjecture remains open, extensive work in the field has led to constant-competitive algorithms for various natural classes of matroids, including uniform \cite{babaioff2007matroids}, graphic \cite{babaioff2009secretary,KorulaP09}, laminar \cite{im2011secretary,jaillet2013advances}, $k$-sparse \cite{soto2013matroid}, and regular matroids \cite{dinitz2014matroid}, as well as in restricted models like random weight assignment \cite{soto2013matroid,santiago2023constant} and even in random assignment with adversarial arrival~\cite{oveis2013variants}.

Interestingly, this separation between constant-competitive algorithms for special classes of matroids and super-constant competitive ratios for the general case is closely tied to which information the algorithm has access to. While the $O(\log \log(\rank(\M)))$-competitive algorithms only require an independence oracle on the arrived elements, some of the constant-competitive algorithms use some additional (upfront) information about the matroid. It is widely conjectured that a constant-competitive algorithm for general matroids should not require any information beyond what can be inferred from an independence oracle on the arrived elements \cite{santiago2023simple,santiago2023constant,FeldmanSZ18}.

A particular natural class of matroids with a gap in the best-known approximation ratio is the class of graphic matroids. In the simpler version of the problem, the algorithm receives edges $(u,v)$ represented by a pair of vertices (or a graph representation of the entire matroid) and can use the vertex information to make decisions. There has been a series of approximation algorithms for this problem: $3e$ by Babaioff et al. \cite{babaioff2007matroids,babaioff2018matroid}, $2e$ by Korula-P\'{a}l \cite{KorulaP09}, $4$ by Soto et al. \cite{soto2021strong} and more recently $3.95$ by Banihashem et al. \cite{BanihashemHKKMO25} and B\'erczi et al.~\cite{10.1007/978-3-031-93112-3_10} who also give a stronger approximation if the graph contains no parallel edges.
In contrast, prior to this work, there was no constant-competitive algorithm for graphic matroids without vertex information.

\paragraph{Graphic MSP without the Graph.} 
In the general matroid secretary problem we access the matroid through an independence oracle, i.e., we can query whether each subset of arrived elements is independent or not. In this paper, we consider the general MSP setting, with the promise that the matroid has a graphic representation. That is, we are promised that there exists a graph $G$ whose forests correspond to independent sets of the given matroid $\M$, but we are not given access to this graph. 
Instead, for any subset of elements (corresponding to a subset of edges of the unknown graph) we can query whether it contains a cycle. 

To get a sense of why this problem is hard, consider three elements such that each subset of them is independent. In this model, it is impossible at this point to know if they form a path, a star or if they are three loose edges (see \Cref{fig:paths-not-id}). If we were given the vertex information, we would know this exactly.  This intuition was extended by Cristi et al.~\cite{cristi2024online} to show that it is impossible to construct a graphic representation in an online manner. In fact, they also show that even constructing a constant-approximate representation online (for a suitable approximability notion) is impossible. This eliminates the hope that we can reduce the independence oracle version of the MSP to the version with vertex information.

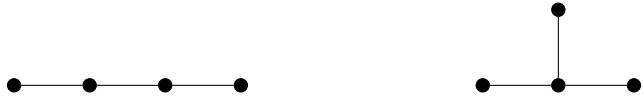
\begin{figure}[h]\centering\begin{tikzpicture}[main_node/.style={circle, draw, thick},unlabeled_node/.style={circle, draw, fill=black,minimum size=5pt, inner sep=0pt},]\node[unlabeled_node] (a) at (0,0) {};\node[unlabeled_node] (b) at (1,0) {};\node[unlabeled_node] (c) at (2,0) {};\node[unlabeled_node] (d) at (3,0) {};\draw (a) -- (b);\draw (b) -- (c);\draw (c) -- (d);\begin{scope}[xshift=6.2cm]\node[unlabeled_node] (a) at (0,0) {};\node[unlabeled_node] (b) at (1,0) {};\node[unlabeled_node] (c) at (2,0) {};\node[unlabeled_node] (d) at (1,1) {};\draw (a) -- (b);\draw (b) -- (c);\draw (b) -- (d);\end{scope}\end{tikzpicture}\caption{This figure demonstrates that the same graphic matroid admits different representations. Hence, the independence oracle doesn't uniquely determine the underlying graph structure.}\label{fig:paths-not-id}
\end{figure}

\subsection{Our Results and Techniques}

Our main result is a polynomial-time constant-approximation for the unknown graphic matroid secretary problem, resolving the open question discussed above. 

\begin{theorem}[Informal]
There exists a randomized $36$-competitive polynomial-time algorithm for the MSP on unknown graphic matroids. The algorithm is ``probability-competitive'', in that each element of the optimal solution is selected with probability at least $1/36$.
\end{theorem}

Our algorithm conceptually splits the input into two parts: the first about $\gamma n$ elements we refer to as the sample $S$, and the remaining roughly $(1-\gamma)n$ elements. It then flips a coin, and uses one of two algorithms to pick an independent set. The first algorithm is designed to pick an independent subset of the sample $S$, and ignore the rest of the input. The second algorithm observes the sample elements $S$ without picking any of them. It uses a graphic representation of the sampled matroid $\M|S$ to produce a graphic representation of (a subset of) the remaining elements in an online manner.

A key conceptual insight of our work is the introduction of a notion of connectivity of an element with respect to the sample set $S$, which drives our entire analysis. At a high level, we consider an element $e$ to be highly connected if the endpoints of $e$ in a graphic representation of $S\cup\{e\}$ are connected by three internally vertex-disjoint paths, not counting $e$ itself (\Cref{def:connectivity}). Crucially, we formulate this connectivity entirely in terms of matroid circuits. This formulation is what allows the algorithm to classify arriving elements on the fly using only independence oracle queries, bridging the gap between graph topology and matroid access.

Our main technical contribution is a structural framework that uses this connectivity dichotomy to overcome the ambiguity of the unknown underlying graph. We show that the first algorithm picks up a constant fraction of the elements in the optimum independent set that are classified as weakly connected. The second algorithm picks up a constant fraction of the optimal elements classified as highly connected by dynamically embedding them into an \emph{arbitrary} valid representation of the sample; we prove that the strict topological constraints of these highly connected elements---which form a \emph{$\theta$-structure}---mathematically force a consistent embedding. In combination, we show that each element of the optimal solution will be selected by one of the two algorithms with sufficient probability, no matter how it is classified.

\paragraph{Dealing with Weakly Connected Elements.} 

Our first algorithm processes elements that are $(\leq 2)$-connected with respect to the sample, and are not parallel to another element in the sample. A natural instinct for these elements is to apply a simple greedy strategy: accept any arriving independent element with a constant probability. However, this approach fails because weak connectivity does not inherently prevent an element from being spanned by the selected sample.
Consider the graph in \Cref{fig:parallel-paths}. The pairs $(u,w)$ and $(w,v)$ are highly connected, but the edge $(u,v)$ is only $1$-connected. If the number of paths between $u-w$ and $w-v$ is sufficiently large, a greedy algorithm that independently selects edges with constant probability will almost certainly select a full path from $u$ to $v$ through $w$. This spans the edge $(u,v)$ early in the process, preventing it from being chosen upon arrival despite its weak global connectivity.

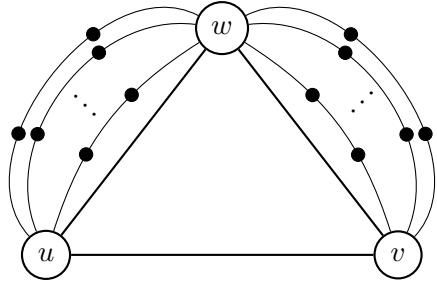
\begin{figure}[t]
\centering
\begin{tikzpicture}[
    main_node/.style={circle, draw, thick},
    unlabeled_node/.style={circle, draw, fill=black,minimum size=5pt, inner sep=0pt},
]

\node[main_node] (u) {$u$};
\node[main_node, right=4cm of u] (v) {$v$};

\node[main_node] at ($(u)!0.5!(v)$) [yshift=3cm] (w) {$w$};

\draw[thick] (u) -- (v);

\draw[thick] (u) -- (w);
\draw[thick] (v) -- (w);

\draw (u) to[bend left=80] node[unlabeled_node, pos=0.33] {} node[unlabeled_node, pos=0.66] {} (w);
\draw (u) to[bend left=60] node[unlabeled_node, pos=0.33] {} node[unlabeled_node, pos=0.66] {} (w);

\path (u) to[bend left=40] node[midway, rotate=135] {$\dots$} (w);

\draw (u) to[bend left=20] node[unlabeled_node, pos=0.33] {} node[unlabeled_node, pos=0.66] {} (w);

\draw (v) to[bend right=80] node[unlabeled_node, pos=0.33] {} node[unlabeled_node, pos=0.66] {} (w);
\draw (v) to[bend right=60] node[unlabeled_node, pos=0.33] {} node[unlabeled_node, pos=0.66] {} (w);

\path (v) to[bend right=40] node[midway, rotate=-135] {$\dots$} (w);

\draw (v) to[bend right=20] node[unlabeled_node, pos=0.33] {} node[unlabeled_node, pos=0.66] {} (w);

\end{tikzpicture}
\caption{A graph with a large number of parallel paths between $u-w$ and $v-w$ that passes through node $w$. Therefore, $u$ and $v$ are $1$ connected in the given graph but the edge $u-v$ is spanned with high probability by a sample which contains each edge independently with a constant probability.}
\label{fig:parallel-paths}
\end{figure}

To overcome the failure of the naive greedy approach, our high-level idea is to exploit the structural obstruction that certifies weak connectivity. For graphic matroids, an element being $(\leq 2)$-connected implies the existence of a $2$-vertex cut. If we can guarantee that the algorithm avoids heavily intersecting this obstruction, we can ensure that the weakly connected edge $e$ remains available upon its arrival. We achieve this by introducing the concept of a \emph{blocked} element (\Cref{def:blocked_new}). An arriving element $e$ is considered blocked if it is either already spanned by the selected elements, or if there exists a previously observed but unselected element $e'$ such that $e$ is spanned by $e'$ combined with the selected elements.

The primary technical challenge is balancing this condition. If the blocking rule is too strict, it successfully protects future elements but excessively rejects currently available ones. We prove that our specific formulation strikes the right balance: it is strong enough to keep future elements available, while ensuring that any weakly connected optimal element is unblocked (and thus selectable) with constant probability when it arrives.

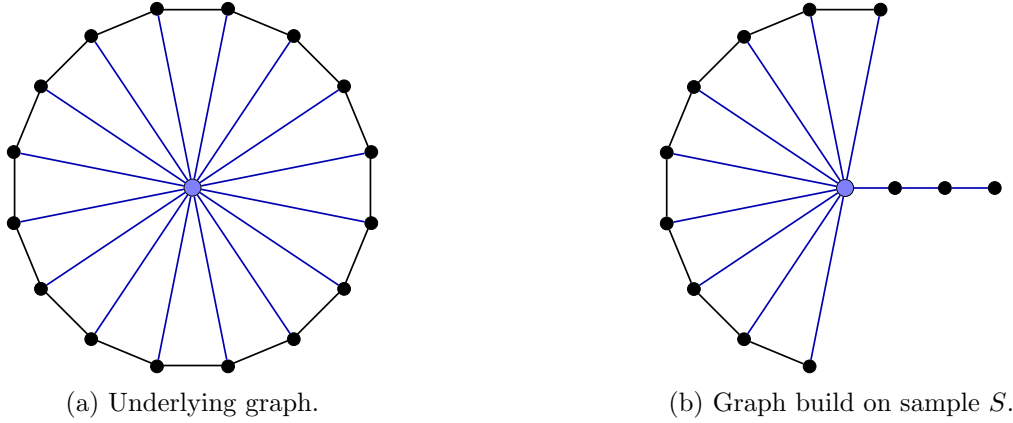
\begin{figure}[t]
    \centering
    \begin{subfigure}[b]{0.48\textwidth}
        \centering
        \scalebox{0.8}{
        \begin{tikzpicture}[
            scale=0.55, 
            vertex_node/.style={circle, draw, fill=black, inner sep=1.5pt, minimum size=6pt},
            center_node/.style={circle, draw, fill=blue!50!white, inner sep=2pt, minimum size=8pt, font=\small, text=white},
            edge_style/.style={thick, black}, 
            center_edge_style/.style={thick, blue!70!black}, 
            polygon_radius/.initial=3cm, 
        ]

        \node[regular polygon, regular polygon sides=16, 
              minimum size=2*\pgfkeysvalueof{/tikz/polygon_radius}, 
              draw, thick, black] (hexadecagon) {};

        \foreach \i in {1,...,16} {
            \node[vertex_node] (v\i) at (hexadecagon.corner \i) {};
        }

        \node[center_node] (center_node) at (0,0) {};  

        \foreach \i in {1,...,16} {
            \draw[center_edge_style] (center_node) -- (v\i);
        }

        \end{tikzpicture}
        }
        \caption{Underlying graph.}
        \label{fig:wheel_full}
    \end{subfigure}%
    \hfill
    \begin{subfigure}[b]{0.48\textwidth}
        \centering
        \scalebox{0.8}{
        \begin{tikzpicture}[
            scale=0.55, 
            vertex_node/.style={circle, draw, fill=black, inner sep=1.5pt, minimum size=6pt},
            center_node/.style={circle, draw, fill=blue!50!white, inner sep=2pt, minimum size=8pt, font=\small, text=white},
            edge_style/.style={thick, black}, 
            center_edge_style/.style={thick, blue!70!black}, 
            polygon_radius/.initial=3cm, 
        ]

        \node[regular polygon, regular polygon sides=16, 
              minimum size=2*\pgfkeysvalueof{/tikz/polygon_radius}] (hexadecagon) {};

        \foreach \i in {1,...,9} {
            \node[vertex_node] (v\i) at (hexadecagon.corner \i) {};
        }

        \node[center_node] (center_node) at (0,0) {};  

        \foreach \i in {1,...,8} {
            \pgfmathtruncatemacro{\nexti}{\i+1}
            \draw[edge_style] (v\i) -- (v\nexti);
        }

        \foreach \i in {1,...,9} {
            \draw[center_edge_style] (center_node) -- (v\i);
        }

        \node[vertex_node] (p1) at (1.5, 0) {};
        \node[vertex_node] (p2) at (3.0, 0) {};
        \node[vertex_node] (p3) at (4.5, 0) {};
        
        \draw[center_edge_style] (center_node) -- (p1) -- (p2) -- (p3);

        \end{tikzpicture}
        }
        \caption{Graph build on sample $S$.}
        \label{fig:wheel_half}
    \end{subfigure}

    \caption{A central node connected to vertices. (a) A full wheel graph making all edges $3$-connected. (b) A subgraph constructed on sample $S$ which might not be consistent with the underlying full graph since some black edges are indistinguishable from blue edges in sample.}
    \label{fig:wheel_comparison}
\end{figure}

\paragraph{Dealing with Highly Connected Elements.} Our second algorithm is designed to collect reward from elements that are $(\ge3)$-connected with respect to the sample set $S$ (or are parallel to some element in $S$). A natural instinct for these elements is to leverage the uniqueness of graphic representations for 3-connected graphs \cite{truemper1980whitney} to reconstruct the true underlying graph. One might hope to build this global representation from the sample and seamlessly slot arriving highly connected elements into their correct vertex endpoints. However, this approach fails because the sample $S$ itself typically lacks global 3-connectivity, making its representation highly ambiguous. Consider the graph in \Cref{fig:wheel_comparison}. Even if the full underlying graph is 3-connected (\Cref{fig:wheel_full}), an online algorithm observing only a random fraction of its edges cannot deduce unobserved structural vertices (like the central vertex). Thus, any valid graphic representation of the sample built on the fly might be severely distorted and inconsistent with the true underlying structure (\Cref{fig:wheel_half}). Embedding future elements into this distorted representation would clash with their true structural dependencies, violating the underlying matroid structure.

To overcome the failure of global graph reconstruction, our high-level idea is to embrace the arbitrary nature of the sample's representation and dynamically embed arriving elements using the matroid's fundamental circuits as a structural guide. Instead of trying to find the true graph, we fix an \emph{arbitrary} valid graphic representation $G(\phi_S)$ and a basis $B$ of the sample. We achieve this embedding by computing the fundamental circuit $C(e,B)$ with respect to $B$ whenever a new element $e \notin S$ arrives. Let $P = C(e,B) \setminus \{e\}$. If $P$ happens to form a simple continuous path in our arbitrarily constructed representation $G(\phi_S)$, we safely extend our graph by assigning $e$ to connect the two endpoints of $P$. We call such an element \emph{$S$-aligned}.

The primary technical challenge is ensuring that this dynamic embedding strategy is both consistent and complete, despite the arbitrary choice of $G(\phi_S)$. If our representation is distorted, one might worry that an element could form circuits with multiple distinct subsets in $S$, leading to conflicting endpoint assignments, or that a highly connected element might fail to form a path entirely and be erroneously discarded. We prove that the structural properties of matroids neatly resolve both concerns. For consistency, we show that the symmetric difference properties of circuits in graphic matroids guarantee that if \emph{one} circuit involving $e$ maps to a path in $G(\phi_S)$, then \emph{all} subsets of $S$ forming a circuit with $e$ will also map to parallel paths sharing the exact same pair of endpoints. 
For completeness, we show that the obstruction certifying 3-connectivity is what we call a $\theta$-structure: three disjoint subsets of elements where the union of any two forms a circuit. The $\theta$-structure forces these interlocking cycles to be drawn as parallel paths, guaranteeing a valid path exists in \emph{any} arbitrarily chosen graphic representation of the sample.
By emitting each $S$-aligned element alongside its uniquely deduced vertex pair, we provide a direct, black-box reduction from the unknown graphic matroid problem on highly connected elements to the known graphic matroid secretary problem.

\paragraph{Random Order CRS for Unknown Graphic Matroids.} As a corollary of our techniques and result, we obtain an $O(1)$-balanced random-order CRS for unknown graphic matroids, strengthening the result from Santiago and Wang~\cite{santiago2023simple}, who proposed a simple $\frac{1}{96}$-balanced random-order CRS when upon arrival of the elements, the endpoints of the edges are also revealed. We defer the technical details to \Cref{sec:OCRS}.

\section{Matroid Preliminaries}

\paragraph{Matroids.}
A matroid $\M = (E, \I)$ 
is defined by a ground set $E$ and a non-empty family of independent sets $\I \subseteq 2^E$ containing $\emptyset$ that satisfy the matroid axioms: 
\begin{itemize}
\item Downward closure: If $X \subseteq Y \subseteq E$ and $Y \in \I$, then $X \in \I$.
\item Exchange property: If $X, Y \in \I$ and $|X| > |Y|$, then there exists an element $x \in X \setminus Y$ such that $Y \cup \{x\} \in \I$. 
\end{itemize}

\paragraph{Graphic Matroids.}
We say that a matroid admits a graphic representation if there exists a set $V$ and a map 
$\phi: E \rightarrow \binom{V}{2}$
such that $S \in \I$ iff $\phi(S)$ is acyclic.
We can think of each element $e \in E$ as an edge in an (undirected) graph with vertices $V$. In this representation, a set of elements (edges) is independent if it is a forest in the graph.

A matroid $\M$ is graphic if it admits a graphic representation. When it exists, the graphic representation is in general not unique. For example, given three elements $\{e_1, e_2, e_3\}$ with $\I = 2^E$, both representations in \Cref{fig:paths-not-id} are valid. 
In particular, a set of elements $P$ that forms a path in some graphic representation of $\M$, is not necessarily a path in every graphic representation of $\M$ ($P$ may not be connected, and/or may contain vertices of degree $>2$ in other representations).

Seymour~\cite{seymour1981recognizing} gave a polynomial-time algorithm to construct a valid graphic representation of a matroid $\M$,
as long as $\M$ is graphic. (The algorithm only needs access to an independence oracle for $\M$.)

\paragraph{Rank, Closure and Circuits.}
The rank of a matroid is a function $\rank:2^E \rightarrow \Z_+$ defined as the size of any maximal independent subset of $X$: $\rank(X) = \max\{|Y| \mid Y \in \I, Y \subseteq X\}$.
The \emph{closure} of a set $X$ is defined as the set of elements whose addition does not increase its rank: 
\[ \cl(X) = \{e\in E \mid \rank(X\cup \{e\}) = \rank(X) \} \]

A set $C$ is a \emph{circuit} if it is minimally dependent: that is, removal of any element from $C$ yields an independent set.
For any subset $X \subseteq E$, let $\Circ(X)$ denote the set of circuits contained in $X$:
\[ \Circ(X) = \{ C \subseteq X \mid C \notin \I \text{ and } S \in \I, \forall S \subsetneq C \}. \]
Note that for graphic matroids, a set $C$ is a circuit in $\M$ if and only if $C$ maps to a cycle in every graphic representation of $\M$.

\paragraph{Parallel Elements.} We say that two elements $e, e'$ in a matroid are parallel if $\rank(\{e,e'\}) = \rank(\{e\}) = \rank(\{e'\}) =1$. In a graphic matroid this corresponds to parallel edges.
If an edge $e$ is parallel to some $e'\in S$, any graphic representation of the sample $S$ is trivially extended to include
$e$ (since $e$ and $e'$ are parallel, $e$ must be incident to the same pair of vertices as $e'$). 
Such edges can be handled by the second algorithm, along with highly connected elements.

\paragraph{Optimal Basis.}
Given a subset $X \subseteq E$ we define $w(X) = \sum_{e \in X} w(e)$. We will use $\Opt(X)$ to denote: 
\[ \Opt(X) = \argmax_Y \{w(Y) \mid Y \subseteq X, Y \in \I \}. \]
If the weights $w$ are all distinct then there is a unique optimum subset for each $X$. If not, we will break ties lexicographically.

\paragraph{Unknown Graphic MSP.}
We consider the matroid secretary problem (MSP) on a graphic matroid, where
we are given access to the matroid structure but not to the representation. More precisely, the algorithm knows $n = \abs{E}$ but knows nothing about the weight of each element or the matroid (other than the promise that it is graphic). The algorithm processes the elements in random order: if $e_1, \hdots, e_n$ is a random permutation of $E$, then at time $t$, the algorithm learns the weight $w(e_t)$ and can query for any $X \subseteq E_t := \{e_1, \hdots, e_t\}$ whether $X \in \I.$ 

At time $t$, the algorithm has the choice to accept element $e_t$ or irrevocably reject it with the constraint that the set of accepted elements must be independent in the matroid. If $I$ is the accepted set, the competitive ratio of the algorithm is given by: 
\[
\frac{\E[w(I)]}{w(\Opt(E))}.
\]

\paragraph{Probability Distributions.} We use $\operatorname{Ber}(p)$ to denote the Bernoulli distribution with parameter $p$, $\operatorname{Bin}(n,p)$ to denote the Binomial distribution with parameters $n,p$, and $\operatorname{Unif}(0,1)$ to denote the uniform distribution over support $(0,1)$.

\section{Connectivity and Element Decomposition}  

Our central idea for approaching the unknown graphic MSP is to estimate the connectivity of each element as it arrives, applying a different algorithm depending on whether it is weakly or highly connected to the rest of the matroid. 
Since we lack access to the entire matroid, we must estimate element connectivity by sampling. Our algorithm for the weakly connected elements will select from the samples, while the algorithm for the highly connected elements will select from the remaining elements.

\paragraph{Samples.} 
A common approach to the MSP is to sample each element independently with probability $\gamma$, to ``learn'' properties of the instance at hand, 
and then use the learned information to make informed decisions on the remaining elements. Since the elements arrive in random order, the natural way to add each $e \in E$ to the sample set independently with probability $\gamma$ is to sample an index $m \sim \operatorname{Bin}(n, \gamma)$
and define the sample set $S = E_{m}$, i.e., the first $m$ elements. 
(Equivalently, one can think of sampling each element independently with probability $\gamma$; these two views coincide under random arrival order.)

\paragraph{Connectivity w.r.t.~Samples.} Consider a graphic matroid. Given sampled elements, we define the connectivity of each element w.r.t.~the samples as follows: for any element $e \in E$ and set of samples $S \subseteq E$, we say that element $e$ is $k$-connected w.r.t.~$S$ if there exist at least $k$ circuits $C_1, \dots, C_k \subseteq S \cup \{e\}$ containing $e$, such that for any distinct $i,j \in [k]$, $C_i \cap C_j = \{e\}$ and $C_i \Delta C_j$ is a circuit. More formally:

\begin{definition}[$k$-connectivity] \label{def:connectivity} 
Consider a graphic matroid. 
Given a subset $S \subseteq E$ and an element $e \in E$, we say that $e$ is $k$-connected with respect to $S$ if there exist circuits $C_1, \dots, C_k \in \Circ(S \cup \{e\})$ such that for any distinct $i,j \in [k]$ it holds that $C_i \cap C_j = \{e\}$ and $C_i \Delta C_j \in \Circ(S)$.
\end{definition}

Note that $e$ is $k$-connected w.r.t.~$S$ if and only if it is $k$-connected w.r.t.~$S\setminus\{e\}$.

Next, we show that our definition of an element $e \in E\setminus S$ being $k$-connected with respect to $S$ is equivalent to the property that in any graphic representation of the matroid there are $k$ internally vertex-disjoint paths in $S$ that connect the two endpoints of $e=(u,v)$ (disjoint except for the two endpoints themselves). The advantage of our seemingly less immediate definition capturing this property is that it can be checked with access to an independence oracle only.

\begin{lemma}\label{lemma:vertex_connectivity}
Consider a graphic matroid. For any integer $k\geq 1$, an element $e$ is $k$-connected with respect to $S$ iff in any graphic representation of the matroid, if $e$ is mapped to edge $(u,v)$, then there are $k$ internally vertex-disjoint paths in $S$ connecting $u$ and $v$.
\end{lemma} 
\begin{proof}
Note that if $e\in C \in \Circ({S \cup \{e\}})$, then in any graphic representation $C$ is a simple cycle of the graph, and hence $C \setminus \{e\}$ is a path connecting the endpoints of $e$. Finally, observe that if $C_i \setminus \{e\}$ and $C_j \setminus \{e\}$ are vertex-disjoint (other than the endpoints of $e$), then $C_i \Delta C_j$ is a circuit (see left side of \Cref{fig:k-connected-oracle}). Conversely, if they have a common vertex $w$ in their graphic representation, then we can decompose $C_i \Delta C_j$ into at least two edge-disjoint cycles (sharing vertex $w$), and hence it properly contains a circuit and is not a minimal dependent set.
\end{proof}

\begin{figure}[t]
\centering
\begin{tikzpicture}[
    main_node/.style={circle, draw, thick},
    unlabeled_node/.style={circle, draw, fill=black,minimum size=5pt, inner sep=0pt},
]
\node[main_node] (u) {$u$};
\node[main_node, right=2cm of u] (v) {$v$};
\foreach \i in {1,2} {
    \node[unlabeled_node, right=1cm of u, yshift=-3cm+\i*2cm] (a\i) {};
}
\path[draw, dashed] (u) edge node[below] {$e$} (v);
\draw (u) to[bend left] (a2);
\draw (a2) to[bend left] (v);
\draw (u) to[bend right] (a1);
\draw (a1) to[bend right] (v);

\begin{scope}[xshift=6.2cm]
\node[main_node] (u) {$u$};
\node[main_node, right=2cm of u] (v) {$v$};
\node[unlabeled_node, right=.9cm of u, yshift=.62cm] (w) {};
\path[draw, dashed] (u) edge node[below] {$e$} (v);
\draw (u) to[out=30, in=90+10]  (v);
\draw (u) to[out=90-10, in=180-30] (v);
\end{scope}
\end{tikzpicture}
\caption{On the left, when the circuits $C_1$ and $C_2$ are internally vertex disjoint, then $C_1 \Delta C_2$ is a circuit. On the right, if $C_1$ and $C_2$ share a vertex then $C_1 \Delta C_2$ is a union of circuits.} 
\label{fig:k-connected-oracle}
\end{figure}
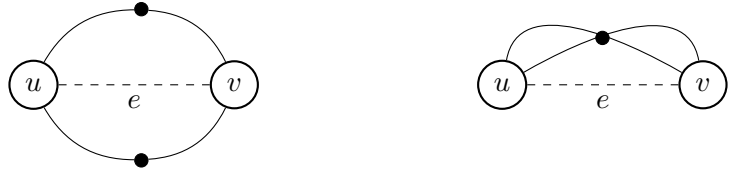

\paragraph{Element Decomposition.} We partition the elements
in $E$ based on their connectivity w.r.t.~the sample: 
\begin{align*}
  \EH &:= \{e \in E \mid \text{$e$ is $3$-connected w.r.t.~$S$ or $e$ is parallel to some $e'\in S$}\}  \\
  \EW &:= E \setminus \EH.
\end{align*}

The elements in $\EH$ are highly connected; those in $\EW$ are weakly connected. We note that the sets $\EH$ and $\EW$ are random as they depend on the sample $S$. In the following sections, we design two algorithms: 
\Cref{alg:weakly-connected} obtains a constant approximation to the optimal weight contained in $\EW \cap S$. 
\Cref{alg:highly-connected-alt} obtains a constant approximation to the optimal weight contained in $\EH \setminus S$. 

Although $\EW \cap S$ and $\EH \setminus S$ do not partition $E$, the uniformly random arrival order ensures they contain a constant fraction of the optimal weight of $\EW$ and $\EH$, respectively. Thus, approximating the optimum on just these two disjoint sets suffices to achieve an overall constant competitive ratio.

\section{Dealing with Weakly Connected Elements}\label{sec:weakly_connected}

Our first algorithm for unknown graphic matroids obtains a constant fraction of the total weight of the elements in $\Opt(E) \cap \EW \cap S$, which are elements in the sample that are $(\leq 2)$-connected. 
(Note: The definition of $\EW$ ensures that $\EW \cap S$ doesn't contain any parallel edges.)
The algorithm is a randomized greedy procedure with a carefully defined ``blocking'' condition such that each element in $\EW \cap S$ remains available with a constant probability.

\subsection{Blocked Elements via Dependence}

Intuitively, we consider an element to be \emph{unblocked} if its addition to the already accepted elements yields an independent set, even if we were to retroactively add any single previously arrived and rejected element. We formalize this as follows.

\begin{definition}[Blocked element]\label{def:blocked_new}
At time $t$, let $E_{t-1}=\{e_1,\dots,e_{t-1}\}$ be the set of elements that arrived before $e_t$, and let $X_{t-1}\subseteq E_{t-1}$ be the set selected by \Cref{alg:weakly-connected} so far.
We say that the arriving element $e_t\notin E_{t-1}$ is \emph{blocked} by $(X_{t-1},E_{t-1})$ if either:
\begin{itemize}
\item $e_t$ is dependent on $X_{t-1}$, that is, $e_t\in \cl(X_{t-1})$; or 
\item there exists an element $e'\in E_{t-1}$ such that $e_t \in \cl(X_{t-1}\cup\{e'\})$.
\end{itemize}
We say that $e_t$ is \emph{unblocked} if it is not blocked.
\end{definition}

Note that this notion of blocking is stronger than just checking if a newly arrived element can be feasibly added to the solution. As we shall argue, this is precisely the obstruction needed for weakly connected elements, enabling a clean union bound argument.

\subsection{Algorithm Description}

\Cref{alg:weakly-connected} processes the sample elements $e_1,\dots,e_m$ (with $m\sim\operatorname{Bin}(n,\gamma)$). At each time $t\le m$, it considers $e_t$ and selects it with probability $\beta$ provided it is unblocked.

\begin{algorithm}
	\begin{algorithmic}[1]
		\Require Independence oracle, $n$ (total number of elements), random arrival order $\pi$ of elements. 
            \State Let $X \gets \emptyset$
            \State Sample $m \sim \operatorname{Bin}(n, \gamma)$
		\ForAll{$t= 1,\dots, m$}
            \State Sample $Z_{t} \sim  \operatorname{Ber}\left( \beta \right )$  
            \If{$Z_t = 1$ \textbf{and} $e_t$ is unblocked by $(X_{t-1},E_{t-1})$}
                \State Select element $e_t$ and update $X \gets X \cup \{e_t\}$
            \EndIf        
        \EndFor
	\end{algorithmic}
	\caption{Secretary Algorithm for Weakly Connected Elements}
    \label{alg:weakly-connected}
\end{algorithm}

We note that \Cref{alg:weakly-connected} trivially returns a feasible set, as it explicitly queries the independence oracle to ensure the selected elements remain independent.

\subsection{Bounding the Probabilities of Selecting \texorpdfstring{$\EW$}{E\_W}-Elements}

We begin by showing that in any graphic representation $\phi^*$ (not known to the algorithm), any two vertices remain disconnected with probability at least $1-\beta$.

\begin{lemma}\label{lemma:nodes-likely-disconnected_new}
Fix any graphic representation $\phi^*$ of $E$ and any two vertices $u,v$ in $\phi^*$. Let $X$ be the set selected by \Cref{alg:weakly-connected}.  
Then
\[
\Pr[\text{$u$ and $v$ are connected in $\phi^*(X)$}] \;\le\; \beta.
\]
\end{lemma}

\begin{proof}
Let $\tau$ be the first time index (if it exists) such that the arriving edge $e_\tau$ would connect $u$ and $v$ in $\phi^*(X_{\tau-1}\cup \{e_\tau\})$. If no such edge ever arrives, $u$ and $v$ are never connected.

Because $\tau$ is the first such time, $e_\tau$ is unblocked. Therefore, $e_\tau$ is selected (and $u$ and $v$ actually become connected) if and only if the coin flip $Z_\tau = 1$. If $Z_\tau = 0$, $e_\tau$ is rejected, and $u$ and $v$ remain disconnected in $\phi^*(X_\tau)$.

Crucially, if $Z_\tau = 0$, $u$ and $v$ will never become connected in the future. Any future edge $e_t$ (for $t > \tau$) that would connect $u$ and $v$ in $\phi^*(X_{t-1}\cup \{e_t\})$ must form a cycle with $e_\tau$ in $\phi^*(X_{t-1}\cup\{e_\tau, e_t\})$. This cycle implies that $e_t \in \cl(X_{t-1}\cup\{e_\tau\})$. Because $e_\tau \in E_{t-1}$, $e_t$ is blocked by $e_\tau$ (see \Cref{fig:blocking_condition}).

In particular, conditioned on the history up to time $\tau$, the event that $u$ and $v$ ever become connected requires that $Z_\tau=1$. Since $Z_\tau$ is independent of the history and $\Pr[Z_\tau=1]=\beta$, we have
\begin{align*}
\Pr[\text{$u$ and $v$ ever become connected}] 
&\;=\; \sum_{t} \Pr[\tau=t]\cdot \Pr[Z_t=1\mid \tau=t] \\
&\;\le\; \sum_t \Pr[\tau=t]\cdot \beta \\
&\;\le\; \beta,
\end{align*}
completing the proof.
\end{proof}

\begin{figure}[h]
\centering
\begin{tikzpicture}[
    node/.style={circle, fill=black, inner sep=2pt},
    main line/.style={thick},
    dashed edge/.style={dashed, thick}
]

    \node[node, label=below:{\Large $u$}] (u) at (0,0) {};
    \node[node] (p1) at (1.2,0) {};
    \node[node] (j1) at (2.4,0) {}; 
    
    \node[node] (j2) at (6.4,0) {}; 
    \node[node] (p2) at (7.6,0) {};
    \node[node, label=above:{\Large $v$}] (v) at (8.8,0) {};

    \draw[main line] (u) -- (p1) -- (j1);
    \draw[main line] (j2) -- (p2) -- (v);

    \node[node] (u_start) at (3.8, 1.2) {}; 
    \node[node] (u_end)   at (5.0, 1.2) {};
    
    \draw[main line] (j1) to[out=90, in=180] (u_start);
    \draw[dashed edge] (u_start) -- node[above=3pt] {\Large $e_t$} (u_end);
    \draw[main line] (u_end) to[out=0, in=90] (j2);

    \node[node] at (2.75, 0.85) {};
    \node[node] at (6.05, 0.85) {};

    \node[node] (l_start) at (3.8, -1.2) {};
    \node[node] (l_end)   at (5.0, -1.2) {};

    \draw[main line] (j1) to[out=-90, in=180] (l_start);
    \draw[dashed edge] (l_start) -- node[below=3pt] {\Large $e_\tau$} (l_end);
    \draw[main line] (l_end) to[out=0, in=-90] (j2);

    \node[node] at (2.75, -0.85) {};
    \node[node] at (6.05, -0.85) {};

\end{tikzpicture}
\caption{Illustration of the blocking condition in the proof of \Cref{lemma:nodes-likely-disconnected_new}.} 
\label{fig:blocking_condition}

\end{figure}
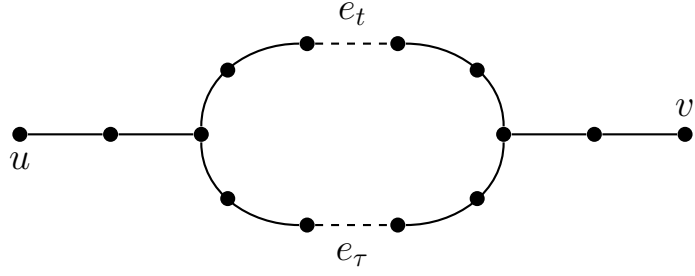

We next establish a link between blocked weakly connected elements, and (small) vertex cuts in any graphic representation of the underlying matroid.

\begin{definition}[Vertex cut]
Consider a representation $\phi:E \rightarrow {V \choose 2}$ and an edge $e$ with endpoints $\phi(e) = \{u,v\}$. We say that $W \subseteq V \setminus \{u,v\}$ is an $e$-\emph{vertex-cut} if after removing vertices $W$ from the graph $\phi(E \setminus \{e\})$, vertices $u$ and $v$ get disconnected.
\end{definition}

\begin{lemma}\label{lemma:blocked_implies_reach_cut}
Fix a graphic representation $\phi^*$ of $E_{t}$, and let $e_t$ have endpoints $\{u,v\}$ under $\phi^*$. Suppose $e_t$ is $(\le 2)$-connected w.r.t.~$E_{t-1}$, and there is no $e'\in E_{t-1}$ parallel to it.
Let $W$ be an $e_t$-vertex-cut in $\phi^*(E_{t-1})$ with $|W|\le 2$. 
If $e_t$ is blocked by $(X_{t-1},E_{t-1})$ (\Cref{def:blocked_new}), then there exists a pair $(x,w)\in \{u,v\}\times W$ such that $x$ and $w$ are connected in $\phi^*(X_{t-1})$.
\end{lemma}

\begin{proof}
If $e_t$ is dependent on $X_{t-1}$, then $\phi^*(X_{t-1})$ contains a $u$--$v$ path. Since $W$ separates $u$ and $v$ in $\phi^*(E_{t-1})$, every $u$--$v$ path intersects $W$, hence some $w\in W$ lies on that path. In particular, either $u$ is connected to $w$ in $\phi^*(X_{t-1})$ or $v$ is connected to $w$ in $\phi^*(X_{t-1})$ (indeed both hold along the path).
If instead $e_t$ is independent of $X_{t-1}$ but there exists $e'\in E_{t-1}$ such that $e_t$ is dependent on $X_{t-1}\cup\{e'\}$, then $\phi^*(X_{t-1}\cup\{e'\})$ contains a $u$--$v$ path. Removing the single edge $e'$ from that path leaves a walk that connects $u$ to one endpoint of $e'$ using only edges in $X_{t-1}$, and connects the other endpoint of $e'$ to $v$ using only edges in $X_{t-1}$. Since $W$ is a vertex cut in $\phi^*(E_{t-1})$, at least one of these two $X_{t-1}$-subpaths must intersect $W$, yielding some $w\in W$ that is connected in $\phi^*(X_{t-1})$ to $u$ or to $v$.
\end{proof}

The previous lemma shows that there are very few specific vertex pairs whose connection could trigger the blocking condition for a weakly connected element. This allows us to bound the overall probability of a weakly connected element being blocked via a simple union bound.

\begin{lemma}\label{lemma:likely-not-blocked_new}
Consider a graphic matroid.
Let $e_t\in \EW\cap S$ be an arriving element processed by \Cref{alg:weakly-connected}.  
Then the probability that $e_t$ is not blocked is at least $1-4\beta$. Hence, the probability that $e_t$ is accepted (selected) is at least $\beta(1-4\beta)$.
\end{lemma}
\begin{proof}
Fix a graphic representation $\phi^*$ of $E_{t}$, and let the endpoints of $e_t$ be $\{u,v\}$ under $\phi^*$. Since $e_t\in \EW$, it is $(\le 2)$-connected w.r.t.~$E_{t-1}$, and thus admits an $e_t$-vertex-cut $W$ with $|W|\le 2$ in $\phi^*(E_{t-1})$.
By \Cref{lemma:blocked_implies_reach_cut}, if $e_t$ is blocked then there exists $(x,w)\in \{u,v\}\times W$ such that $x$ and $w$ are connected in $\phi^*(X_{t-1})$. There are at most $|\{u,v\}\times W|\le 4$ such pairs. By \Cref{lemma:nodes-likely-disconnected_new}, each such connectivity event has probability at most $\beta$. By a union bound,
\[
\Pr[\text{$e_t$ is not blocked}] \;\ge\; 1-4\beta.
\]
Finally, $Z_t$ is independent of the history up to time $t$ and of the event ``$e_t$ is not blocked'', hence
\[
\Pr[\text{$e_t$ is selected}] \;\ge\; \Pr[Z_t=1]\cdot \Pr[\text{$e_t$ is not blocked}] \;\ge\; \beta(1-4\beta).\qedhere
\]
\end{proof}

Finally, because every optimal weakly connected element in the sample is selected with this constant probability, linearity of expectation immediately yields our desired approximation guarantee.

\begin{lemma}\label{lem:alg_2_cond}
Consider a graphic matroid, and let $X$ be the set of elements returned by \Cref{alg:weakly-connected}. For any fixed sample $S$,
$$
\E[w(X) \mid S] \;\ge\; \beta(1-4\beta)\cdot w\!\left(\Opt(E)\cap \EW \cap S\right).
$$
\end{lemma}
\begin{proof}
Consider any edge $e\in \Opt(E)\cap \EW \cap S$. When it arrives, it is processed by \Cref{alg:weakly-connected}, and (by \Cref{lemma:likely-not-blocked_new}) it is selected with probability at least $\beta(1-4\beta)$. Linearity of expectation yields the claim.
\end{proof}

Note that we have actually shown something stronger: conditioned on $S$, each element $e\in S\cap \EW$ and hence each element $e \in \Opt(E)\cap \EW \cap S$ is selected with constant probability.

\section{Dealing with Highly Connected Elements}\label{sec:highly_connected}

We next present and discuss our algorithm for unknown graphic matroids that obtains a constant fraction of the total weight of the elements in $\Opt(E) \cap (\EH \setminus S)$, i.e., the highly connected elements. 

The matroid $\M$ is graphic, and so is its restriction $\M|S$ to the sample $S$. A graphic representation of $\M|S$ can be constructed using an algorithm of Seymour~\cite{seymour1981recognizing} in polynomial time, given access to an independence oracle for $\M|S$.
We show that this graphic representation can be extended to a graphic representation of $\M|(S\cup A)$ which contains $\EH$.
Given this graphic representation of $\M|(S\cup A)$, we present a black-box reduction to \emph{known} graphic MSP.

\subsection{Constructing a Graphic Representation of \texorpdfstring{$\EH$}{E\_H}} 

Let $\phi_S : S \rightarrow \binom{V}{2}$ be a graphic representation of the matroid $\M|S$ on the sample elements $S$.

We will use $G(\phi_S)=(V, \{(e, u, v) \mid (u,v)=\phi_S(e), e\in S\})$ to denote the graph defined by the representation $\phi_S$.

Let $e\in E\setminus S$ be an element such that there is a set $P\subseteq S$ such that $P\cup \{e\}$ is a circuit in $\M$, and $P$ maps to a path in the graphic representation $\phi_S$. We call such an element \emph{$S$-aligned}.

\begin{lemma}\label{lemma:circuit-path}
Consider an element $e$ and two subsets $P, Q \subseteq S$ such that both $P\cup \{e\}$ and $Q\cup \{e\}$ are circuits in $\M$.
Assume $P$ is a path in the graphic representation $\phi_S$. Then $Q$ is also a path in $\phi_S$, and $P$ and $Q$ are parallel 
(i.e., $P$ and $Q$ connect the same pair of endpoints).
\end{lemma}

\begin{proof}
$\M$ is graphic, and therefore also a binary matroid.
The Circuit Characterization of binary matroids (see, e.g., \cite[Theorem 9.1.2]{oxley2011matroid}) says the symmetric difference of any two circuits is a union of disjoint circuits. Hence,
$P\Delta Q = (P\cup \{e\}) \Delta (Q\cup \{e\})$ is the disjoint union of circuits of $\M$. 
Since $P\Delta Q$ is fully contained in $S$, it must also correspond to a disjoint union of cycles in the graph $\phi_S$. 

$P$ is a path in $G(\phi_S)$. That means the subgraph of $G(\phi_S)$ spanned by $P$ has exactly two vertices of odd degree. Let those vertices be $u, v$.
All vertex degrees in the subgraph spanned by $P\Delta Q$ are even. Hence, $Q=(P\Delta Q)\Delta P$ also has exactly two odd degree vertices, $u, v$.
Since $Q$ is an independent set in $\M$, it is acyclic in the graph $G(\phi_S)$. An acyclic graph with only two odd-degree vertices is necessarily a path. 
Hence, $Q$ is a path connecting $u$ and $v$ in the graph $\phi_S$.
\end{proof}

Let $A$ be the set of all $S$-aligned elements $e\in E\setminus S$.

Consider a mapping $\phi : A \rightarrow \binom{V}{2}$ constructed as follows. 
Let $e\in A$. Since $e$ is $S$-aligned, let $P$ be a path in $\phi_S$ such that $P\cup \{e\}$ is a circuit in $\M$. 
Let $u, v$ be the endpoints (degree-1 vertices) of $P$ in $\phi_S$. Define $\phi(e) = \{u, v\}$.
\Cref{lemma:circuit-path} shows that $\phi(e)$ is well defined, in that all paths that form a circuit with $e$ have the same pair of endpoints.

Let $\M|(S\cup A) = \M\setminus (E\setminus (S\cup A))$ denote the matroid obtained from $\M$ by restricting it to elements of $S\cup A$.

\begin{lemma}\label{lemma:representation-valid}
The mapping $\phi_{S\cup A} = \phi_S \cup \phi$ is a valid graphic representation of the matroid $\M|(S\cup A)$.
\end{lemma}

\begin{proof}
$\phi_S$ is a graphic representation of $\M|S$.
$\M|(S\cup A) = \M\setminus (E\setminus S\setminus A)$ is a matroid extension of $\M|S$.
We wish to show that $\phi_{S\cup A} = \phi_S \cup \phi$ is a graphic representation of $\M|(S\cup A)$.
We proceed in three steps. 

First, we construct a binary matrix representation $\psi_S$ of the matroid $\M|S$ from its graphic representation $\phi_S$. 
A matrix representation $\psi_S : S \rightarrow \mathbb{F}_2^d$ maps each element to a $d$-dimensional binary vector, for some $d$.
We choose $d=|V|$. 
That is, $\psi_S(e)$ will be a vector with a coordinate for each vertex $v\in V$.
For each element $e$, let $\phi_S(e) = \{u, v\}$ for some $u, v\in V$.
$\psi_S(e)(w)=1$ if $w\in\{u, v\}$ and $\psi_S(e)(w)=0$ otherwise.
By a standard argument, $\psi_S$ is a valid binary matrix representation of $\M|S$.

Second, a binary matrix representation $\psi_S$ of a matroid $\M|S$ uniquely extends to any extension of the matroid that lies in its closure.
Let $B$ be some basis of $\M|S$, and let $e\notin S$ be some element such that $e\in \cl(S)$. Let $C \subseteq B \cup \{e\}$ be the fundamental circuit containing $e$. 
For each $e\in A$, define $\psi(e) = \bigoplus_{e'\in C\setminus \{e\}} \psi_S(e')$. (This is possible, as $A$ lies in the closure of $S$.)
By a standard argument, $\psi_{S\cup A}$ is a matrix representation of the binary matroid $\M|(S\cup A)$.

Third, note that for each $e\in A$, its fundamental circuit $C$ (with $e$ removed) is a path in $\phi_S$, due to \Cref{lemma:circuit-path}.
Hence, $\psi_{S\cup A}(e)$ is a vector with exactly two non-zero coordinates. Let's denote the nonzero coordinates $u, v$ and define
$\phi(e) = \{u, v\}$.

Since $\psi_{S\cup A}$  is a valid binary matrix representation of $\M|(S\cup A)$, $\phi_{S\cup A}$ must be a valid graphic representation of the same.
\end{proof}

\paragraph{The $\theta$-structure.}
Finally, we need to argue that $\EH\setminus S$ is contained in $A$.
Consider an element $e\in \EH\setminus S$. \Cref{def:connectivity} says $e$ is 3-connected w.r.t.~$S$ if there are three disjoint sets $P_1, P_2, P_3\subseteq S$ such that $P_i\cup\{e\}$ is a circuit, and $P_i\cup P_j$ is a circuit for $i\ne j$.
The structure $P_1\cup P_2 \cup P_3$ is known as a \emph{$\theta$-structure} or \emph{$\theta$-matroid} in the literature. 

It is well known that every $\theta$-matroid is graphic, and that in any graphic representation, the sets $P_1, P_2, P_3$ are paths, and these paths are internally disjoint and parallel (that is, connect the same pair of endpoints).
For completeness, we provide a proof for the following narrower lemma that is sufficient for our analysis.

\begin{lemma}\label{lemma:theta-matroid-path}
If $e\notin S$ is 3-connected in $S$, then for any graphic representation $\phi_S$ of $\M|S$, 
there is a path $P$ in $\phi_S$ such that $\{e\}\cup P$ is a circuit in $\M$.
\end{lemma}

Our proof of the lemma makes use of the following claim.

\begin{claim}\label{claim:theta-matroid-noextracircuits}
Let $P_1, P_2, P_3$ be three disjoint subsets of $S$ such that $P_i\cup\{e\}$ is a circuit, and $P_i\cup P_j$ is a circuit for $1\le i <j \le 3$.
Then, $\M|(P_1\cup P_2\cup P_3)$ contains no circuits besides $P_i\cup P_j$.
\end{claim}

\begin{proof}
Since $\M$ is graphic, the restriction $\M|(P_1\cup P_2\cup P_3)$ is also graphic. In any graphic representation of this restriction, the circuits $P_1\cup P_2$, $P_2\cup P_3$, and $P_1\cup P_3$ imply that $P_1$, $P_2$, and $P_3$ are internally disjoint paths connecting the same pair of endpoints. The union of these paths forms a $\theta$-graph, in which the only cycles are the unions of two paths $P_i\cup P_j$. Since circuits in a graphic matroid correspond to simple cycles in its graphic representation, there are no other circuits in the restriction. 
\end{proof}

\begin{proof}[Proof of \Cref{lemma:theta-matroid-path}]
Consider some graphic representation $\phi_S$ of a matroid $\M|S$ where $P_1 \cup P_2 \cup P_3 \subseteq S$.
Since $P_i\cup P_j$ is a circuit, it must correspond to a simple cycle in the graph $\phi_S$. 

Observe that since $P_1$ is a subgraph of the simple cycle $P_1\cup P_2$, it must be either a path, or a collection of path segments in $\phi_S$. 
Let $V'$ be the set of endpoints of these path segments. (Each vertex $v\in V'$ is incident to exactly one path segment: if multiple path segments were incident to $v$, they would cease to be path segments).
By the same token,  $P_2$ must also be a collection of path segments, with endpoints being exactly the vertices of $V'$.
Applying the same reasoning to the cycle $P_1\cup P_3$ yields that each of $P_1, P_2, P_3$ consists of $|V'|/2$ path segments, with each path segment being incident to a pair of vertices in $V'$, and each vertex $v\in V'$ being incident to exactly one path segment from each $P_i$.

If $|V'| = 2$, then $P_1$ is a path, so \Cref{lemma:theta-matroid-path} holds.
Now consider $|V'| \ge 4$. Since each $P_i$ contains $|V'|/2$ path segments, there are $\frac{3}{2}|V'|$ path segments total.
Since there are more path segments than vertices, there must be a simple cycle that doesn't visit all path segments. 

Consider the following procedure. Start at an arbitrary vertex $v_0\in V'$. In each step $i$, pick the unique path segment from $P_{(i \bmod 3) + 1}$ that is incident to $v_i$, and denote the other endpoint of the path segment $v_{i+1}$. Stop when you arrive at a vertex for the second time. 

This procedure yields a closed walk that excludes at least one path segment from at least two of the sets $P_1, P_2, P_3$. 
This closed walk must contain a simple cycle $C$ that does not fully contain (and hence is not equal to) any of $(P_1\cup P_2)$, $(P_2\cup P_3)$ or $(P_3\cup P_1)$. 
This contradicts \Cref{claim:theta-matroid-noextracircuits}.
\end{proof}

\begin{algorithm}[H]
    \begin{algorithmic}[1]
        \Require Independence oracle, $n$ (total number of elements), stream of elements of $\M$
        \Ensure Vertex set $V$. Stream of $n-m$ elements, each either labeled ``dummy'', or annotated with a vertex pair.
            \State Sample $m\sim \operatorname{Bin}\left(n, \gamma \right)$.
            \State Add first $m$ arrived edges into sample $S$.
            \State Build a graph representation $\phi_S$ over $S$ using Seymour's algorithm. 
            \State Pick a basis $B$ of the matroid $\M|S$ (which is a spanning forest in $G(\phi_S)$).
            \State Emit $V$, the vertex set of the graphic representation $\phi_S$.
            \For{each arriving edge $e$}
                \State If $e\notin \cl(S)$, emit ($e$, ``dummy'') and continue
                \State Else, let $C=C(e, B)$ be the fundamental circuit of $e$, and let $P=C\setminus \{e\}$
                \State If $P$ is not a path in $G(\phi_S)$, emit ($e$, ``dummy'') and continue
                \State Else, let $u,v$ be the endpoints of $P$. Emit ($e, \{u, v\}$).
            \EndFor
    \end{algorithmic}
    \caption{Algorithm for Highly Connected Edges}
    \label{alg:highly-connected-alt}
\end{algorithm}

\begin{theorem}\label{thm:representation_construction} 
\Cref{alg:highly-connected-alt} produces a stream of $n-m$ elements of $E\setminus S$ such that:
\begin{itemize}
    \item Vertex pair annotations on emitted elements not marked ``dummy'' produce a valid graphic representation of the matroid consisting of non-dummy emitted elements.
    \item All elements of $\EH\setminus S$ are emitted as non-dummy.
\end{itemize}
\end{theorem}

\begin{proof}
The first property is a direct consequence of \Cref{lemma:representation-valid}.
For the second property, consider an $e\in \EH\setminus S$. If $e$ is $3$-connected in $S$, we appeal to \Cref{lemma:theta-matroid-path}. Otherwise, there must be an $e'\in S$ that is parallel to $e$. But this means that $\{e, e'\}$ is a circuit and $e'$ is a path in $S$, and thus $e\in A$ is not dummy.
\end{proof}

\subsection{Reduction to Graphic MSP with Known Graph}
The output of \Cref{alg:highly-connected-alt} is a vertex set, followed by a sequence of edges (each annotated with the vertex pair incident to it, and its weight), that must be immediately accepted or rejected. 

This is almost exactly the input to the graphic matroid secretary problem where the graph is known, with one exception: \Cref{alg:highly-connected-alt} may emit dummy edges. 
To be able to use existing algorithms for graphic MSP (e.g., \cite{KorulaP09}), we use the following reduction.

\begin{algorithm}[H]
    \begin{algorithmic}[1]
        \Require Vertex set $V$.
        Stream of $n-m$ elements, each either labeled ``dummy'', or annotated with a pair of vertices from $V$.
        \Ensure Vertex set $V'$. A stream of $n-m$ edges, each annotated with a vertex pair and weight.
        \State Let $V''$ be a set of dummy vertices of sufficient size so that ${|V''| \choose 2} \ge n-m$.
        \State Let $\pi$ be a random permutation of the elements of ${V'' \choose 2}$.
        \State Emit $V'=V\cup V''$ as the new vertex set
        \For{each arriving element $e$}
            \State If $e$ is dummy, emit $e$, the next available vertex pair from $\pi$, and weight $0$
            \State If $e$ is not dummy, emit $e, u, v$, and the original weight $w(e)$.
        \EndFor
    \end{algorithmic}
    \caption{Handling dummy elements}
    \label{alg:highly-connected-reduction}
\end{algorithm}

\begin{lemma}\label{lem:dummy_reduction}
Suppose \Cref{alg:highly-connected-reduction} gets as its input a random permutation of the union of the edges of some graph $G=(V, E)$ and a set of dummy elements $E''$.
Then, its output is a random permutation of edges of a (random) graph $G'=(V', E')$ such that $G'$ is a supergraph of $G$, and $E'\setminus E$ consists of edges of weight $0$.
\end{lemma}

\begin{proof}
By construction, $G'$ is a supergraph of $G$, and edges in $E'\setminus E = E''$ have weight $0$.
What's left is to show that the edges of $G'$ arrive according to a uniformly random permutation. 

Consider a run of \Cref{alg:highly-connected-reduction} that resulted in emitting some graph $G'$.
By construction, a random subset (determined by $\pi$) of vertex pairs from ${V''\choose 2}$ of size $|E''|$ got included in $G'$.
Conditioned on the set of vertex pairs selected, the order in which the pairs show in the output is uniformly random, by virtue of the corresponding prefix of $\pi$ being a uniformly random permutation of $|E''|$ elements.
\end{proof}

\begin{lemma}\label{lemma:alg1_analysis_cond_alt} 
Consider a graphic matroid, and let $X$ be the set of elements returned by \Cref{alg:highly-connected-alt}. Conditioned on the sample $S$, 
\[
\E[w(X) \mid S] \;\geq\; \frac{1}{4} \cdot w\!\left(\Opt(E) \cap (\EH \setminus S)\right).
\]
\end{lemma}
\begin{proof}
We establish the bound by proving a stronger claim: each individual element $e\in \Opt(E) \cap (\EH \setminus S)$ is selected by the algorithm with probability at least $1/4$. 

The algorithm of \cite{BanihashemHKKMO25} is guaranteed to select each edge $e\in \Opt(G_\pi)$ with probability at least $\frac{1}{3.95} \geq \frac{1}{4}$. By the optimality of the greedy algorithm on matroids, since $A \subseteq E$, any element $e\in \Opt(E) \cap A$ must also belong to $\Opt(A)$. By construction, the real edges of $G_\pi$ are exactly $A$, and the dummy edges have weight $0$, so $\Opt(G_\pi) \cap A = \Opt(A)$. Since $\EH \setminus S \subseteq A$, any element $e \in \Opt(E) \cap (\EH \setminus S)$ is also in $\Opt(E) \cap A$, which implies $e \in \Opt(A) \subseteq \Opt(G_\pi)$. Therefore, \mbox{$\Pr[e \in X \mid S] \geq 1/4$}. Applying linearity of expectation over all elements in \mbox{$\Opt(E) \cap (\EH \setminus S)$} yields the claim.
\end{proof}

If the instance $\M$ is known to be simple (i.e. not contain parallel elements/edges), 
a subroutine that takes advantage of this fact can be used instead, to obtain a tighter competitive guarantee.  
B\'erczi et al.~\cite{10.1007/978-3-031-93112-3_10} give an algorithm for the known-representation graphic MSP that is $0.2693$ probability-competitive.

\section{Putting It All Together and Poly-Time Implementation}

We now combine the two algorithms to obtain the main result of the paper for graphic matroids. Let $F^*=\Opt(E)$ be an optimal solution. We define two random subsets of $F^*$ based on the sample $S$:
\[ 
F_H^* \;=\; F^* \cap (\EH\setminus S),
\qquad
F_W^* \;=\; F^* \cap (\EW \cap S).
\]
\Cref{lemma:alg1_analysis_cond_alt,lem:alg_2_cond} state that, conditioned on $S$, \Cref{alg:highly-connected-alt} obtains a $\rho_H=\frac{1}{4}$-approximation to $w(F_H^*)$ and \Cref{alg:weakly-connected} obtains a $\rho_W=\beta(1-4\beta)$-approximation to $w(F_W^*)$.

\begin{theorem}\label{thm:main}
Randomizing between \Cref{alg:weakly-connected} and \Cref{alg:highly-connected-alt} with probabilities $\gamma$ and $1-\gamma$, respectively, and choosing $\beta=1/8$ and $\gamma=2/3$ yields a polynomial-time randomized algorithm for the MSP on unknown graphic matroids that is $36$-competitive.
\end{theorem}

\begin{proof}
Fix an element $e\in F^*$. Let $q_e := \Pr[e\in \EW]$ (with respect to the randomness of the sample $S$). Note that the event $e\in S$ is independent of the event $e\in \EW$ because $\EW$ is determined entirely by $S\setminus\{e\}$. Therefore,
\begin{align*}
&\Pr[e\in F_W^*] = \Pr[e\in S]\Pr[e\in \EW] = \gamma \cdot q_e,
\qquad\qquad\qquad\text{and}\\
&\Pr[e\in F_H^*] = \Pr[e\notin S]\Pr[e\in \EH] = (1-\gamma) \cdot (1-q_e).
\end{align*}

By \Cref{lemma:alg1_analysis_cond_alt,lem:alg_2_cond} and linearity of expectation,
\begin{align*}
\E[w(\Alg)] 
&\ge (1-\gamma)\cdot \rho_H\cdot \E[w(F_H^*)] \;+\; \gamma\cdot \rho_W\cdot \E[w(F_W^*)] \\
&= \sum_{e\in F^*} w(e)\Big(\rho_H(1-\gamma)^2(1-q_e) + \rho_W\gamma^2 q_e\Big)\\
&\ge \min\{\rho_H(1-\gamma)^2,\, \rho_W\gamma^2\}\cdot \sum_{e\in F^*} w(e)\\
&= \min\{\rho_H(1-\gamma)^2,\, \rho_W\gamma^2\}\cdot w(F^*).
\end{align*}
With $\rho_H=\frac{1}{4}$, $\beta=\frac{1}{8}$ (so $\rho_W=\beta(1-4\beta)=\frac{1}{16}$), and $\gamma=\frac{2}{3}$, we have $\rho_H(1-\gamma)^2 = \frac{1}{4}(\frac{1}{9}) = \frac{1}{36}$ and $\rho_W\gamma^2 = \frac{1}{16}(\frac{4}{9}) = \frac{1}{36}$. The minimum above evaluates to exactly $\frac{1}{36}$. Hence, $\E[w(\Alg)]\ge \frac{1}{36} \cdot w(\Opt(E))$.
\end{proof}

We remark that while \Cref{thm:main} is stated in terms of expected weight, our analysis proves a stronger probabilistic claim: every element in the optimal solution is chosen by the algorithm with probability at least $1/36$.

\paragraph{Polynomial-Time Implementation.} 
We conclude by arguing that both \Cref{alg:weakly-connected} and \Cref{alg:highly-connected-alt} can be implemented using a polynomial number of independence oracle queries and polynomial computation time. Our primary tool is the result by Seymour~\cite{seymour1981recognizing}, which provides an algorithm to construct a graphic representation of a graphic matroid using polynomially many independence oracle queries.

First, consider \Cref{alg:weakly-connected}, which processes the first $m$ arriving elements $e_1, \dots, e_m$. At each stage $t$, the algorithm must check two conditions for $e_t$:
\begin{enumerate}
\item $X_{t-1} \cup \{e_t\} \in \I$: This requires a single independence oracle query.
\item $e_t$ is unblocked by $X_{t-1}$: This can be efficiently verified by iterating over all elements $e'\in E_{t-1}$ and querying the oracle to ensure that $e_t$ is independent of $X_{t-1}\cup\{e'\}$. This requires at most $t-1$ queries.
\end{enumerate}

Next, consider \Cref{alg:highly-connected-alt}. The algorithm begins by sampling the set $S$ using the first $m \sim \operatorname{Bin}(n, \gamma)$ elements. It then applies Seymour's algorithm~\cite{seymour1981recognizing} to construct the graphic representation $\phi_S$ over $S$. We then pick a basis $B$ of $\M|S$. 
For each subsequent arriving element $e$, we can implement the steps of \Cref{alg:highly-connected-alt} efficiently:
\begin{itemize}
    \item We check if $e \in \cl(S)$ by querying if $B \cup \{e\}$ is dependent. This requires a single independence oracle call.
    \item If $e \in \cl(S)$, we find the fundamental circuit $C(e, B)$ in $\M|(S \cup \{e\})$ using at most $|B|$ independence oracle queries (by standard circuit extraction).
    \item We then check if $P = C(e, B) \setminus \{e\}$ is a path in the graph $G(\phi_S)$. Since $G(\phi_S)$ is explicitly constructed, this is a purely graph-theoretic check and requires no oracle queries.
    \item If $P$ is a path, we map $e$ to the endpoints of $P$ in $G(\phi_S)$ in $O(1)$ time.
\end{itemize}
Thus, the online phase of \Cref{alg:highly-connected-alt} is extremely efficient and does not require running Seymour's algorithm or max-flow computations online. The subsequent reduction to handle dummy elements (\Cref{alg:highly-connected-reduction}) and the selection procedure using the algorithm of Banihashem et al.~\cite{BanihashemHKKMO25} also run in polynomial time. 

Since both algorithms process each element using operations that run in $\operatorname{poly}(n)$ time, the total number of oracle queries and the computational time for the combined algorithm are polynomial in $n$.

\section{Conclusion}

The graphic matroid secretary problem with an unknown graph remained an open problem for at least a decade and motivated the development of techniques that highlighted barriers for natural classes of algorithm to be constant competitive for the secretary problem~\cite{bahrani2021formal}. Until this paper, the unknown graphic matroid secretary problem remained one of the candidate classes for a super-constant lower bound for MSP.

The natural question is whether the techniques in this paper offer a way forward for the general problem. While many barriers still remain, it hints that connectivity with respect to a sample is an useful tool for partitioning the matroid elements into classes that can be handled with different techniques. While the notion of connectivity we define may appear heavily dependent on the graphic representation, it is a version of general notion of Tutte-connectivity functions that is studied in matroid theory (with the only difference that we define Tutte-connectivity with respect to a sample). See~\cite{tutte1966connectivity} or Chapter 8 of~\cite{oxley2011matroid}. In fact, our weakly connected part of the proof can be written entirely in purely-matroid terms. We use vertex covers in the proof for simplicity of presentation, but we could equivalently re-write the analysis in terms of Tutte-separators. We formally develop this idea and partial extension in \Cref{sec:beyond_graphic}. The strongly connected section, however, depends deeply on a graphic representation.

Nevertheless, the success of this structural dichotomy demonstrates that sample-based connectivity is a powerful algorithmic lens. We are optimistic that leveraging the purely matroidal analog could serve as a crucial stepping stone toward resolving the unknown matroid secretary problem for more general classes of matroids (such as regular, binary, and linear).

\bibliographystyle{alpha} 
\bibliography{refs}

@inproceedings{KorulaP09,
  author       = {Nitish Korula and
                  Martin P{\'{a}}l},
  title        = {Algorithms for Secretary Problems on Graphs and Hypergraphs},
  booktitle    = {International Colloquium on Automata, Languages and Programming ({ICALP} 2009)},
  pages        = {508--520},
  year         = {2009},
}

@inproceedings{santiago2023simple,
  title={Simple random order contention resolution for graphic matroids with almost no prior information},
  author={Santiago, Richard and Sergeev, Ivan and Zenklusen, Rico},
  booktitle={Symposium on Simplicity in Algorithms (SOSA 2023)},
  pages={84--95},
  year={2023},
  organization={SIAM}
}

@article{oveis2013variants,
  title={On variants of the matroid secretary problem},
  author={Oveis Gharan, Shayan and Vondr{\'a}k, Jan},
  journal={Algorithmica},
  volume={67},
  number={4},
  pages={472--497},
  year={2013},
  publisher={Springer}
}

@article{truemper1980whitney,
  title={On Whitney's 2-isomorphism theorem for graphs},
  author={Truemper, Klaus},
  journal={Journal of Graph Theory},
  volume={4},
  number={1},
  pages={43--49},
  year={1980},
  publisher={Wiley Online Library}
}

@inproceedings{santiago2023constant,
  title={Constant-competitiveness for random assignment matroid secretary without knowing the matroid},
  author={Santiago, Richard and Sergeev, Ivan and Zenklusen, Rico},
  booktitle={International Conference on Integer Programming and Combinatorial Optimization (IPCO 2023)},
  pages={423--437},
  year={2023},
  organization={Springer}
}

@article{seymour1981recognizing,
  title={Recognizing graphic matroids},
  author={Seymour, Paul D.},
  journal={Combinatorica},
  volume={1},
  number={1},
  pages={75--78},
  year={1981},
  publisher={Springer}
}

@article{soto2021strong,
  title={Strong algorithms for the ordinal matroid secretary problem},
  author={Soto, Jos{\'e} A and Turkieltaub, Abner and Verdugo, Victor},
  journal={Mathematics of Operations Research},
  volume={46},
  number={2},
  pages={642--673},
  year={2021},
  publisher={INFORMS}
}

@article{babaioff2018matroid,
  title={Matroid secretary problems},
  author={Babaioff, Moshe and Immorlica, Nicole and Kempe, David and Kleinberg, Robert},
  journal=jacm,
  volume={65},
  number={6},
  pages={1--26},
  year={2018}
}

@inproceedings{babaioff2009secretary,
  title={Secretary problems: weights and discounts},
  author={Babaioff, Moshe and Dinitz, Michael and Gupta, Anupam and Immorlica, Nicole and Talwar, Kunal},
  booktitle={ACM-SIAM Symposium on Discrete Algorithms (SODA 2009)},
  pages={1245--1254},
  year={2009},
}

@inproceedings{BanihashemHKKMO25,
  author       = {Kiarash Banihashem and
                  MohammadTaghi Hajiaghayi and
                  Dariusz R. Kowalski and
                  Piotr Krysta and
                  Danny Mittal and
                  Jan Olkowski},
  title        = {Beating Competitive Ratio 4 for Graphic Matroid Secretary},
  booktitle    = {European Symposium on Algorithms (ESA 2025)},
  series       = {LIPIcs},
  volume       = {351},
  pages        = {52:1--52:16},
  year         = {2025},
}

@inproceedings{babaioff2007matroids,
  title={Matroids, secretary problems, and online mechanisms},
  author={Babaioff, Moshe and Immorlica, Nicole and Kleinberg, Robert},
  booktitle={ACM-SIAM Symposium on Discrete Algorithms (SODA 2007)},
  pages={434--443},
  year={2007}
}

@inproceedings{cristi2024online,
  title={Online Matroid Embeddings},
  author={Cristi, Andr{\'e}s and D{\"u}tting, Paul and Kleinberg, Robert and Paes Leme, Renato and Patel, Neel},
  booktitle={International Colloquium on Automata, Languages and Programming  (ICALP 2026)},
  year={2026},
  note = {Forthcoming}
}

@INPROCEEDINGS{Lachish14,
  author={Lachish, Oded},
  booktitle={Symposium on Foundations of Computer Science (FOCS 2014)}, 
  title={{O}(log log Rank) Competitive Ratio for the Matroid Secretary Problem}, 
  year={2014},
  pages={326-335},
}

@article{FeldmanSZ18,
  author       = {Moran Feldman and
                  Ola Svensson and
                  Rico Zenklusen},
  title        = {A Simple \emph{O}(log log(rank))-Competitive Algorithm for the Matroid
                  Secretary Problem},
  journal      = {Mathematics of Operations Reserch},
  volume       = {43},
  number       = {2},
  pages        = {638--650},
  year         = {2018},
}

@inproceedings{bahrani2021formal,
  title={Formal barriers to simple algorithms for the matroid secretary problem},
  author={Bahrani, Maryam and Beyhaghi, Hedyeh and Singla, Sahil and Weinberg, S Matthew},
  booktitle={International Conference on Web and Internet Economics (WINE 2021)},
  pages={280--298},
  year={2021},
  organization={Springer}
}

@article{soto2013matroid,
  title={Matroid secretary problem in the random-assignment model},
  author={Soto, Jos{\'e} A},
  journal={SIAM Journal on Computing},
  volume={42},
  number={1},
  pages={178--211},
  year={2013},
  publisher={SIAM}
}

@article{dinitz2014matroid,
  title={Matroid secretary for regular and decomposable matroids},
  author={Dinitz, Michael and Kortsarz, Guy},
  journal={SIAM Journal on Computing},
  volume={43},
  number={5},
  pages={1807--1830},
  year={2014},
}

@inproceedings{im2011secretary,
  title={Secretary problems: {L}aminar matroid and interval scheduling},
  author={Im, Sungjin and Wang, Yajun},
  booktitle={ACM-SIAM Sympoisum on Discrete Algorithms (SODA 2011)},
  pages={1265--1274},
  year={2011},
}

@inproceedings{jaillet2013advances,
  title={Advances on matroid secretary problems: {F}ree order model and laminar case},
  author={Jaillet, Patrick and Soto, Jos{\'e} A and Zenklusen, Rico},
  booktitle={International Conference on Integer Programming and Combinatorial Optimization (IPCO 2013)},
  pages={254--265},
  year={2013},
}

@book{oxley2011matroid,
  title={Matroid theory},
  author={Oxley, James},
  volume={3},
  year={2011},
  publisher={Oxford University Press, USA}
}

@InProceedings{10.1007/978-3-031-93112-3_10,
author="B{\'e}rczi, Krist{\'o}f
and Livanos, Vasilis
and Soto, Jos{\'e} A.
and Verdugo, Victor",
title="Matroid Secretary via Labeling Schemes",
booktitle="Integer Programming and Combinatorial Optimization (IPCO 2025)",
year="2025",
pages="128--141",
isbn="978-3-031-93112-3"
}

@article{tutte1966connectivity,
  title={Connectivity in matroids},
  author={Tutte, William T},
  journal={Canadian Journal of Mathematics},
  volume={18},
  pages={1301--1324},
  year={1966},
  publisher={Cambridge University Press}
}

\appendix

\section*{AI and LLM Disclosure}
We used Artificial Intelligence (AI) and Large Language Models (LLMs) in the preparation of this manuscript for editorial assistance, including polishing, rephrasing, and improving the clarity of parts of the introduction and related work. We also used advanced models for mathematical proofreading. All mathematical claims have been verified by the authors. No citations or BibTeX entries were generated using generative AI.

\section{Random Order OCRS for Unknown Graphic Matroids}\label{sec:OCRS}

In this section, we apply our technique to extend the result of Santiago and Wang~\cite{santiago2023simple} and provide a constant-competitive random-order OCRS for unknown graphic matroids without knowing the endpoints.

\paragraph{Random Order OCRS Setting.}
As a corollary of our techniques and results, we obtain a constant \emph{balanced} random-order CRS for unknown graphic matroids, strengthening the result of Santiago and Wang~\cite{santiago2023simple}, who proposed a $\frac{1}{96}$-balanced random-order CRS when upon arrival of the elements, the endpoints of the edges are also revealed. Their algorithm heavily relies on the information of the endpoints of the edges and does not extend to the case when the endpoints are not revealed. Hence, the more general version of the problem was left as an open problem in \cite{santiago2023simple}.

In OCRS, we are given a vector $\vec x\in \P_{\I}$ feasible in the relaxation of the underlying graphic matroid. Each element $e \in E$ is either active, with probability $x_e$, or inactive, with probability $1 - x_e$, independently of the other elements.

The elements arrive in a random order and whether an element is active is revealed upon its arrival. We assume that the graphic matroid and $\vec x$ are unknown upfront. However, upon arrival of the element $e$, along with its active or inactive status, we learn the corresponding fraction $x_e$. Whenever an arriving element $e \in E$ is active, the algorithm must decide irrevocably whether to select it or not. At all times, the set of currently selected elements $X$ must satisfy $X\in \I$. The goal here is to design a $c$-balanced random-order \emph{contention resolution scheme} which selects each active element with probability $\geq c$.

\paragraph{OCRS Construction.}
Since we can identify the endpoints of edges in $\EH\setminus S$ upon their arrival, we can obtain a constant-competitive random-order CRS for these edges as a corollary of the random-order CRS from \cite{santiago2023simple}. 
\begin{lemma}\label{lem:CRS_EH}
    Conditioned on the sample $S$, there exists a $\frac{1}{96}$-balanced random-order CRS for edges in $\EH\setminus S$ with unknown $\vec x\in \P_{\I}$ which is revealed online. 
\end{lemma}
\begin{proof}
    Due to the $\frac{1}{96}$-balanced random-order CRS for known graphic matroids \cite{santiago2023simple}, when the edges reveal their endpoints in the underlying graphic representation of the matroid, we can resolve contentions. By \Cref{thm:representation_construction}, we can identify the endpoints of edges in $\EH \setminus S$ upon their arrival, allowing us to reconstruct their representation online and apply the OCRS.
\end{proof}
Combining the above lemma with \Cref{alg:weakly-connected}, we obtain a random-order OCRS for unknown graphic matroids with unknown $\vec x$ revealed online.

\begin{theorem}\label{thm:crs}
There exists a $\frac{1}{287}$-balanced random-order CRS for unknown graphic matroids with unknown $\vec x\in \P_{\I}$ which is revealed online, even when the endpoints of edges are not revealed upon arrival.
\end{theorem}
\begin{proof}
We use two schemes and randomize between them:
\begin{itemize}
\item For edges in $\EH\setminus S$, we can identify endpoints online and apply the $\frac{1}{96}$-balanced random-order CRS from \cite{santiago2023simple} (\Cref{lem:CRS_EH}).
\item For edges in $\EW\cap S$, we run \Cref{alg:weakly-connected} treating inactive edges as never eligible; \Cref{lemma:likely-not-blocked_new} implies that each active edge in $\EW\cap S$ is selected with probability at least $\rho_W=\beta(1-4\beta)$.
\end{itemize}
Each edge belongs to $S$ with probability $\gamma$ (independent of its classification into $\EW$ or $\EH$), hence the per-edge guarantee is at least
\[
\min\left\{p'\cdot \frac{1-\gamma}{96},\; (1-p')\cdot \gamma\cdot \rho_W\right\}.
\]
Choosing $\beta=\frac{1}{8}$ (so $\rho_W=\frac{1}{16}$), $\gamma=0.632$, and $p'=0.912$ equalizes the two terms and yields a balance of at least $\frac{1}{287}$.
\end{proof}

\section{Partial Extension to General and Linear Matroids}\label{sec:beyond_graphic}

The weakly connected algorithm (\Cref{alg:weakly-connected}) only needs an independence oracle, and thus can be run on any abstract matroid $\M$. In \Cref{sec:weakly_connected} we analyzed its performance on graphic matroids. In this section, we provide partial results on the algorithm's performance in matroid settings beyond graphic matroids.

First, we prove a fundamental lemma showing that the probability of any element entering the closure of the selected set is bounded. This is the abstract matroid equivalent of \Cref{lemma:nodes-likely-disconnected_new}.

\begin{claim}
In a matroid $M$, let $X\subseteq X'$ be independent sets, and let $e, e'$ and $f$ be distinct elements. Assume that
$X\cup\{e\}$,  $X'\cup\{f\}$ and $X'\cup\{e'\}$ are independent sets, but
$X\cup\{e, f\}$ and $X'\cup\{e', f\}$ are dependent.
Then, $X'\cup\{e', e\}$ is a dependent set.
\end{claim}
\begin{proof}
First, we show that $e \notin X'$. Assume for contradiction that $e \in X'$. Since $X'$ is independent, the subset $X \cup \{e\} \subseteq X'$ is also independent. Since $X \cup \{e, f\}$ is dependent, we must have $f \in \cl(X \cup \{e\}) \subseteq \cl(X')$. This contradicts the assumption that $X' \cup \{f\}$ is independent. Thus, $e \notin X'$. Symmetrically, since $X' \cup \{e'\}$ is independent, we have $e' \notin X'$. Since $e, e'$ are distinct, the set $X' \cup \{e, e'\}$ has size $|X'| + 2$.

Next, since $X \cup \{e, f\}$ is dependent and $X \cup \{e\}$ is independent, we have $e \in \cl(X \cup \{f\})$. Since $X \subseteq X'$, we have $e \in \cl(X' \cup \{f\})$. Symmetrically, since $X' \cup \{e', f\}$ is dependent and $X' \cup \{e'\}$ is independent, we have $e' \in \cl(X' \cup \{f\})$.
Therefore, the set $X' \cup \{e, e'\}$ is contained in $\cl(X' \cup \{f\})$, which implies:
\[
\rk(X' \cup \{e, e'\}) \le \rk(X' \cup \{f\}) = |X'| + 1,
\]
where the equality holds because $X' \cup \{f\}$ is independent. Since the rank of $X' \cup \{e, e'\}$ is at most $|X'| + 1$ but its size is $|X'| + 2$, the set $X' \cup \{e, e'\}$ must be dependent.
\end{proof}

\begin{lemma}
\label{lemma:probability_spanned}
Let $\hat M$ be any matroid extension of $M$, and let $f\in \hat M$ be any fixed element.
Let $X_{t-1}$ be the random set of
elements selected by \Cref{alg:weakly-connected} up to time $t-1$. Then
$$\Pr[f \in \cl_{\hat M}(X_{t-1})] \le \beta.$$
\end{lemma}
\begin{proof}
Consider the first time $t$ when $f\in \cl(X_{t-1} \cup \{e_t\})$. Observe that \Cref{alg:weakly-connected} will select $e_t$ with probability $\beta$.
Conditioned on $e_t$ not getting selected, we claim that $f\notin \cl(X_{t'})$ will hold for all $t'=t, \dots, m$.

To see this, consider a time $t' > t$ when an element $e_{t'}$ arrives such that $f\in \cl(X_{t'-1} \cup \{e_{t'}\})$.
Applying the above claim with $X=X_{t-1}$, $X'=X_{t'-1}$, $e=e_t$, $e'=e_{t'}$, and $f=f$, we obtain that
$e_{t'}$ is dependent on $X_{t'-1} \cup\{e_t\}$, and hence $e_{t'}$ is blocked and will not be selected by \Cref{alg:weakly-connected}.
\end{proof}

\subsection{Tutte Connectivity and \texorpdfstring{$k$}{k}-Separations in Matroids}

Our analysis in \Cref{sec:weakly_connected} relied on a definition of connectivity (\Cref{def:connectivity}) tailored to graphic matroids. To extend the analysis to general and linear matroids, we generalize vertex cuts using Tutte's connectivity function.

\begin{definition}[Tutte Connectivity and $k$-Separation]\label{def:tutte_separation}
Let $M=(E,\I)$ be a matroid, $S\subseteq E$, and $e \in E$. Let $N = M|(S\setminus \{e\})$. 
For any partition $(A, B)$ of $S\setminus \{e\}$, the \emph{Tutte connectivity} of $(A, B)$ in $N$ is defined by
\[
\lambda_N(A) := \rk_M(A) + \rk_M(B) - \rk_M(A \cup B).
\]
Following standard matroid literature, a partition $(A, B)$ of $S\setminus \{e\}$ is a \emph{$k$-separation separating $e$ in $S$} (or a \emph{$k$-separation for $e$ w.r.t.~$S$}) if:
\begin{enumerate}
    \item $e \notin \cl_M(A)$ and $e \notin \cl_M(B)$, and
    \item $\lambda_N(A) \le k - 1$.
\end{enumerate}
We say that $e$ is \emph{$k$-separable} with respect to $S$ if it admits a $k$-separation in $S$.
\end{definition}

\begin{remark}
Under this definition, an element is $1$-separable if it admits a partition $(A, B)$ with $\lambda_N(A) \le 0$. Since $\lambda_N(A) \ge 0$ by submodularity of the rank function, a $1$-separation corresponds exactly to $\rk(A) + \rk(B) = \rk(A \cup B)$.

Note also that if $e$ has a parallel copy $p \in S$ (i.e., $e \in \cl_M(\{p\})$), then $e$ cannot admit any separation $(A, B)$, since $p \in A \implies e \in \cl_M(A)$ and $p \in B \implies e \in \cl_M(B)$. In \Cref{alg:weakly-connected}, any element parallel to a previously arrived element in $S$ is immediately blocked or spanned with probability $1$, so we can assume without loss of generality that $e$ has no parallel copies in $S$.
\end{remark}

\subsection{Performance for 1-Separable Elements in General Matroids}

For elements that are 1-separable, we can prove a performance guarantee for general matroids without any representability assumptions.

\begin{lemma}\label{lemma:1separable_unblocked}
Let $M=(E,\I)$ be a matroid, and $e\in E$. Let $(A, B)$ be a 1-separation separating $e$ in $S = E \setminus \{e\}$ (i.e., $\rk(A) + \rk(B) = \rk(A\cup B)$ with $e\notin \cl(A)$ and $e\notin \cl(B)$).
If we run \Cref{alg:weakly-connected} on $S$, the element $e$ is neither spanned by the selected set $X$ nor blocked by $(X, S)$ with probability at least $1-2\beta$.
\end{lemma}
\begin{proof}
Let $X_A = X\cap A$ and $X_B = X\cap B$ be the elements selected from $A$ and $B$, respectively.
Since $\rk(A) + \rk(B) = \rk(A\cup B)$, the submatroid $M|S$ is the direct sum $M|A \oplus M|B$.
Consequently, the execution of \Cref{alg:weakly-connected} on $S$ decomposes completely into two independent executions on $M|A$ and $M|B$. Specifically, the selection of $X_A$ only depends on elements in $A$ and their coin flips, and similarly for $X_B$. Thus, $X_A$ and $X_B$ are identical to the outputs of running the algorithm on $M|A$ (which is isomorphic to $M/B|_A$) and $M|B$ (isomorphic to $M/A|_B$) independently.

Now consider the contracted matroid $M/B$. In this matroid, the ground set is $A \cup \{e\}$ (since $e \notin \cl(B)$). The algorithm runs on $A$ and selects $X_A$. By \Cref{lemma:probability_spanned} applied to the extension $(M/B)|(A \cup \{e\})$, the probability that $e$ enters the closure of $X_A$ in $M/B$ is at most $\beta$:
\[
\Pr[e \in \cl_{M/B}(X_A)] = \Pr[e \in \cl_M(X_A \cup B)] \le \beta.
\]
Symmetrically, by considering $M/A$ and the execution on $B$, we have:
\[
\Pr[e \in \cl_M(X_B \cup A)] \le \beta.
\]

Suppose $e$ is blocked by $X$ and $S$. By definition of blocking, there exists $e'\in S$ such that $e\in \cl(X\cup \{e'\})$.
Since $S = A \cup B$, we have either $e'\in A$ or $e'\in B$.
\begin{itemize}
    \item If $e'\in A$, then $e\in \cl(X_A \cup X_B \cup \{e'\}) \subseteq \cl(A \cup X_B)$ (since $X_A \cup \{e'\} \subseteq A$).
    \item If $e'\in B$, then $e\in \cl(X_A \cup X_B \cup \{e'\}) \subseteq \cl(X_A \cup B)$ (since $X_B \cup \{e'\} \subseteq B$).
\end{itemize}
Also, if $e$ is spanned by $X$, then $e \in \cl(X) = \cl(X_A \cup X_B)$, which is contained in both $\cl(A \cup X_B)$ and $\cl(X_A \cup B)$.

Therefore, the event that $e$ is spanned or blocked is contained in the union of the events $e\in \cl_M(X_B\cup A)$ and $e\in \cl_M(X_A\cup B)$. By union bound, this occurs with probability at most $2\beta$. Thus, $e$ is neither spanned nor blocked with probability at least $1-2\beta$.
\end{proof}

\begin{theorem}\label{thm:weakly_linear_k1}
Let $e_t$ be the arriving element at step $t$ which is 1-separable w.r.t.~$E_{t-1}$.
The probability that $e_t$ is accepted by \Cref{alg:weakly-connected} is at least $\beta(1 - 2\beta)$ (assuming no parallel elements to $e_t$ in $E_{t-1}$).
For $\beta = 1/4$, this probability is at least $1/8$.
\end{theorem}
\begin{proof}
Since the coin flip $Z_t$ is independent of the history:
\[
\Pr[\text{$e_t$ is selected}] = \Pr[Z_t=1]\cdot \Pr[\text{$e_t$ is not blocked}] \ge \beta(1 - 2\beta),
\]
where we used \Cref{lemma:1separable_unblocked} to bound the probability of being spanned or blocked (with $S = E_{t-1}$).
Setting $\beta = 1/4$ maximizes $\beta(1-2\beta)$ to $1/8$.
\end{proof}

\subsection{Performance for \texorpdfstring{$k$}{k}-Separable Elements in Linear Matroids}

For linear matroids representable over a finite field $\F$, we can generalize the result to any constant connectivity $k$, with a loss depending on the field size. We first define the bottleneck flat for general $k$, which generalizes vertex cuts.

\begin{definition}[$k$-bottleneck flat]\label{def:kbottleneck}
Let $M = (E, \I)$ be a matroid and $e \notin S \subseteq E$. Let $\hat M \supseteq M$ be an extension and $\loops(\hat M):=\cl_{\hat M}(\varnothing)$ its loops.
A flat $K$ of $\hat M$ is a \emph{$k$-bottleneck flat} for $e$ w.r.t.~$S$ if $\rk_{\hat M}(K)\le k+1$ and, for every independent set $X\subseteq S$ and $e'\in S$, if $e\in \cl_{\hat M}(X\cup\{e'\})$ and $e\notin \cl_{\hat M}(\{e'\})$, then $\cl_{\hat M}(X)\cap\bigl(K\setminus \loops(\hat M)\bigr)\neq \varnothing$.
\end{definition}

For linear matroids, the existence of such a flat of small rank allows us to bound the blocking probability by a function of the field size.

\begin{lemma}[Existence of the Bottleneck Flat in Linear Matroids]\label{thm:one-flat-linear}
Let $M=(E,\I)$ be a linear matroid over field $\F$, let $S\subseteq E$ with $e\notin S$, and let $(A,B)$ be a $k$-separation separating $e$ in $S$.
Then $M|(S \cup \{e\})$ has an $\F$-representable extension $\hat M$ and a flat $K$ of rank $\rk_{\hat M}(K) \le k+1$ such that $K$ is a $k$-bottleneck flat for $e$ w.r.t.~$S$.
\end{lemma}
\begin{proof}
Consider a vector representation of $M|(S \cup \{e\})$ over $\F$. Let $U_A:=\langle A\rangle$, $U_B:=\langle B\rangle$, and $\ell_e:=\langle e\rangle$.
Let $W:=(U_A+\ell_e)\cap (U_B+\ell_e)$. By the dimension formula, $\dim W \le k+1$, since $\lambda(A) = \rk_M(A) + \rk_M(B) - \rk_M(A\cup B) \le k - 1$.

For each witness pair $(X,e')$ with $X\subseteq S, e'\in S$ where $e\in \cl(X\cup\{e'\})$ and $e\notin \cl(\{e'\})$, we can construct a non-zero vector $w_{X,e'}\in W\cap \langle X\rangle$ (if $e'\in A$, then $e = u_A + x_B$ for $u_A\in U_A, x_B\in\langle X_B\rangle$, so we set $w_{X,e'} = x_B \in W \cap \langle X \rangle$). We add a new element $k_{X,e'}$ represented by $w_{X,e'}$.
Let $\hat M$ be the resulting linear extension, and let $K := \cl_{\hat M}(\{k_{X,e'}\})$ be the flat generated by these new elements.
Since all generators lie in $W$, $\rk_{\hat M}(K)\le \dim W \le k+1$.
The bottleneck property follows by construction: any independent set $X$ and element $e'$ spanning $e$ will span one of the constructed vectors in $K$.
\end{proof}

\begin{claim}\label{claim:flat_size_linear}
Let $M$ be a linear matroid representable over a finite field $\F$. Any flat $K$ of rank $r \le k+1$ in $M$ contains at most $|\F|^{k+1}$ elements (up to parallel copies).
\end{claim}
\begin{proof}
In any linear representation of $M$ over $\F$, the elements of the flat $K$ are represented by non-zero vectors in an $r$-dimensional linear subspace. Since an $r$-dimensional subspace over $\F$ contains $|\F|^r$ vectors in total, the number of distinct rank-$1$ subspaces (and thus non-parallel elements in $K$) is at most $\frac{|\F|^r - 1}{|\F| - 1} \le |\F|^r \le |\F|^{k+1}$.
\end{proof}

\begin{lemma}\label{lem:weakly_general_unblock}
For a linear matroid $M$ over a finite field $\F$, let $e_t$ be $k$-separable w.r.t.~$E_{t-1}$. The probability that $e_t$ is spanned or blocked by $X_{t-1}$ is at most $\beta |\F|^{k+1}$.
\end{lemma}
\begin{proof}
If $e_t$ is spanned or blocked by $X_{t-1}$, then either $e_t \in \cl(X_{t-1})$ or there exists $e' \in E_{t-1}$ such that $e_t \in \cl(X_{t-1} \cup \{e'\})$. In either case, by the bottleneck flat property (\Cref{thm:one-flat-linear}), $X_{t-1}$ must span a non-loop element of $K$, i.e., $\cl(X_{t-1}) \cap (K \setminus \loops) \neq \emptyset$.
By \Cref{lemma:probability_spanned} and a union bound over the non-loop elements of $K$, this happens with probability at most $\beta |K|$.
By \Cref{claim:flat_size_linear}, since $K$ is a flat of rank $r \le k+1$ in a linear matroid over $\F$, $|K| \le |\F|^{k+1}$.
Thus, the probability that $e_t$ is spanned or blocked is at most $\beta |\F|^{k+1}$.
\end{proof}

\begin{theorem}\label{thm:weakly_linear_field}
Consider a linear matroid over a finite field $\F$.
Let $e_t$ be an arriving element at step $t$ which is $k$-separable w.r.t.~$E_{t-1}$.
The probability that $e_t$ is accepted by \Cref{alg:weakly-connected} for $\beta = \frac{1}{2|\F|^{k+1}}$ is at least $\frac{1}{4|\F|^{k+1}}$.
\end{theorem}

\end{document}